\documentclass{article}
\usepackage{amssymb}
\usepackage{amsmath}
\usepackage{amsthm}
\usepackage{color}
\usepackage{fullpage}
\usepackage{float}
\usepackage{graphicx}
\usepackage{framed}
\usepackage{color,soul}
\usepackage{authblk}
\usepackage{ifthen}
\newcommand{\YSHI}[1]{\ifthenelse{\equal{\draft}{1}}{{\color{red}{#1}}}{#1}}

\usepackage{comment}

\newtheorem{theorem}{Theorem}
\newtheorem{proposition}[theorem]{Proposition}
\newtheorem{definition}[theorem]{Definition}
\newtheorem{lemma}[theorem]{Lemma}
\newtheorem{corollary}[theorem]{Corollary}

\usepackage{graphics}

\newcommand{\commentout}[1]{}

\def\Tr{\textnormal{Tr}}

\def\A{\textnormal{Alice}}
\def\B{\textnormal{Bob}}

\newcommand{\bra}[1]{\langle#1|}
\newcommand{\ket}[1]{|#1\rangle}

\newcommand{\ketbra}[2]{|{#1}\rangle\!\langle{#2}|}

\newcommand{\bbI}{\mathbb{I}}

\newcommand{\calX}{\mathcal{X}}
\newcommand{\calY}{\mathcal{Y}}
\newcommand{\calA}{\mathcal{A}}
\newcommand{\calB}{\mathcal{B}}

\begin{document}

\title{Local Randomness: Examples and Application\footnote{This work was supported by NSF grant 1526928.}}

\author[1]{Honghao Fu}
\author[1,2]{Carl A. Miller}
\affil[1]{\small Department of Computer Science, Institute for Advanced Computer Studies and
Joint Institute for Quantum Information and Computer Science, University of Maryland, College Park, MD, 20740}
\affil[2]{\small National Institute of Standards and Technology,
100 Bureau Dr., Gaithersbug, MD 20899, USA  }

\maketitle

\begin{abstract}
When two players achieve a superclassical score at a nonlocal game, their
outputs must contain intrinsic randomness.  This fact has many useful
implications for quantum cryptography.  Recently it has been observed
(C. Miller, Y. Shi, Quant. Inf. \& Comp. 17, pp. 0595-0610, 2017)  that such scores also imply
the existence of \textit{local randomness} --- that is, randomness known to one player but not
to the other.  This has potential implications for cryptographic tasks between two cooperating
but mistrustful players.  In the current paper we bring this notion toward practical realization,
by offering near-optimal bounds on local randomness for the CHSH game, and also proving the security of
a cryptographic application of local randomness (single-bit certified deletion).
\end{abstract}

\vskip0.2in

Device-independent quantum cryptography \cite{Ekert91, Mayers98} is based on the observation that any Bell inequality
violation guarantees the existence of intrinsic randomness.  In particular, the outputs of such an inequality are known to be unpredictable to an arbitrary adversary.  Work in this field over more than a decade has culminated in recent proofs
of security for quantum key distribution and randomness expansion that are immune to any errors in quantum hardware \cite{Vazirani:dice, Vazirani:QKD:PRL, Miller:2016, Universal-spot, Dupuis:2016, Arnon:2016}. 

It has more recently been observed \cite{Miller2017} that when two spatially separated parties violate a Bell inequality, then the outputs of either player must contain some unpredictability to the other player.  Whereas global randomness (randomness possessed by both parties) is useful in cryptographic tasks in which two players are cooperating, local randomness (randomness possessed by one party and unknown to the other) is potentially useful in cryptographic settings where the parties are interacting but do not trust one another.  This invites an exploration of quantum cryptographic
protocols that are immunized both against imperfections in the quantum hardware and (possibly coordinated) cheating by one of the players. 

Suppose that a nonlocal game $G$ with complete support\footnote{A nonlocal game $G$ has complete
support if the input distribution is nonzero on all elements of  $\mathcal{A} \times \mathcal{B}$.} is played by two players, Alice and Bob, where Alice's input and output alphabets
are $\mathcal{A}$ and $\mathcal{X}$, respectively, and Bob's input and output alphabets are $\mathcal{B}$
and $\mathcal{Y}$, respectively.  A referee chooses an input pair $(a, b)$ according to a fixed distribution
and distributes $a$ to Alice and $b$ to Bob, who return $x$ and $y$ respectively.
The results of \cite{Miller2017} assert that if the expected score of
Alice and Bob's strategy exceeds the best possible classical score by $\epsilon$,
then Bob will not be able to guess Alice's output with probability better than $(1 - \Omega_G (\epsilon^2))$,
even if he were given Alice's input.  In other words, the pair $(a, x)$ is necessarily more random
to Bob than the input letter $a$ alone.  This is an example of \textit{blind} randomness expansion, where the word
``blind'' is used because one player is blind to the randomness generated by the other.
(This can be compared to the notion of ``bound randomness'' in the three-party setting of \cite{Acin:2016}.)

The results of \cite{Miller2017} are highly general but numerically weak.  The goals of the current paper
are (1) to demonstrate techniques that prove numerically strong bounds on local randomness,
and (2) to demonstrate the power of local randomness by proving security for a specific
application (one-shot certified deletion).  Our study is focused on two example games, the CHSH game and the
Magic Square game.  

Section~\ref{sec:pre} reviews some necessary background and then Section~\ref{subsec:npa} outlines the Navascues-Pironio-Acin (NPA) hierarchy \cite{Navascues:2008}, which has been previously used to prove lower bounds on global randomness \cite{pironio10}.
 The key
difference in the case of local randomness is that we must bound the behavior of a party (Bob) who is making two \textit{sequential} measurements on a single system, rather than a single measurements on two separated systems as in the case of global randomness. Fortunately, the NPA hierarchy can be adapted to handle sequential measurements, as observed in \cite{budroni2013bounding, pironio2010convergent}.
Using such an adapted approach, we compute a function $F$ such that any superclassical score of $s$ at the CHSH game
guarantees that Bob cannot recover Alice's output with probability greater than $F ( s )$.  The function $F$ that we obtain is shown to be optimal within a margin of $0.02$.  (See Figure~\ref{sandwichplot}).

A downside of the CHSH game is that, even when a perfectly optimal strategy is used by Alice and Bob,
Bob still has approximately an $85\%$ chance of guessing Alice's output bit.  For some cryptographic
purposes it is more useful for the player to have a bit that approximates a perfect coin flip.  In Section~\ref{sec:rigid}
we study the Magic Square game.  This game is large enough that
is computationally difficult to apply the methods from Section~\ref{sec:bound}, and
so instead we apply the notion of quantum \textit{rigidity}, which asserts that certain nonlocal games have unique winning strategies.  It was recently shown that the Magic Square game \cite{WuBancal:2016_p} is rigid.
We build off of the proof in \cite{WuBancal:2016_p} to show that in any strategy for Magic Square which achieves
an expected score of $1 - \epsilon$, Alice obtains a bit that Bob cannot guess with probability
greater than $1/2 + O ( \sqrt{\epsilon} )$.  (See Corollary~\ref{cor:magicguess}.)

Lastly, in Section~\ref{sec:cert} we provide an initial application of device-independent local randomness by showing
that it enables \textit{single-bit certified deletion}.  In this cryptographic problem,
Bob possesses an encrypted bit \fbox{$m$} which could be read with a key, $k$, possessed only by Alice, and the goal is for Alice and Bob to interact through classical communication
only so that Bob can certifiably delete his copy of \fbox{$m$}.  The resulting deleted
state must be unreadable even if
Bob were to later learn $k$.
We prove that any multi-use device  that performs well at the Magic Square game can be used
for certified deletion.  A formal statement is given in Theorem~\ref{thm:cert}.  Roughly,
the probability that Bob can recover the bit $m$ after deletion is shown to be no more than $\frac{1}{2} + O ( \sqrt{\epsilon})$,
where $\epsilon$ denotes the average probability that the device loses the Magic Square game,
and the probability that Bob can recover $m$ before deletion is $1 - O ( \epsilon )$.

Our result can be compared to other cryptographic tasks for mistrustful parties in the device-independent setting.  Coin-flipping and bit commitment have been proven in the device-independent setting 
\cite{Silman:2011, Aharon:2011, Aharon:2016} with constant (rather than vanishing) bias.  Also, strong
cryptographic primitives have been proven under additional assumptions such as limited quantum storage \cite{Kaniewski:2016, Rib:2016, Ribeiro2016} and relativistic assumptions \cite{Adlam:2015}.  Exploring the upper limits
of device-independence in the mistrustful setting appears to be an interesting open problem.


\section{Preliminaries}
\label{sec:pre}
In this section, we introduce the concepts that formally define nonlocal games and related notations used through out this paper, starting with the definition of a $2$-player correlation.

Our notation follows \cite{Miller2017}.    Let $\mathcal{A,B}$  denote Alice's and Bob's input alphabets, respectively,
and let $\mathcal{X,Y}$ denote Alice's and Bob's output alphabets.
A \textit{$2$-player (input-output) correlation} is a vector
$(P(xy|ab))$ of nonnegative reals, where $(x,y,a,b)$ varies over $\mathcal{X} \times \mathcal{Y}
\times \mathcal{A} \times \mathcal{B}$, such that 
\begin{eqnarray*}
\sum_{xy} P(xy|ab)  = 1
\end{eqnarray*}
for all pairs $(a, b)$, and such that the quantities
\begin{eqnarray*}
\begin{array}{ccc}
P(x|a) := \sum_y P(xy|ab), & \hskip0.1in &
P(y|b) := \sum_x P(xy|ab)
\end{array}
\end{eqnarray*}
are independent of $b$ and $a$, respectively.  (The latter conditions are referred
to as the ``non-signaling'' constraints.)

A $2$-player game is a pair $(q, H)$ where
\begin{eqnarray}
q \colon \mathcal{A} \times \mathcal{B} \to [0, 1 ]
\end{eqnarray}
is a probability distribution and
\begin{eqnarray}
H \colon \mathcal{A} \times \mathcal{B} \times \mathcal{X}
\times \mathcal{Y} \to [0, 1]
\end{eqnarray}
is a function. If $q(a, b)\ne 0$ for all $a\in \mathcal{A}$ and $b\in\mathcal{B}$,
the game is said to have a {\em complete support}. The expected score associated to
such a game for a $2$-player correlation $(P(xy|ab))$
is
\begin{eqnarray}
\sum_{a, b, x, y }  q ( a, b ) H ( a, b, x, y ) P(xy|ab).
\end{eqnarray}

A \textit{$2$-player strategy} is a $5$-tuple
\begin{eqnarray}\label{eqn:strategy}
\Gamma & = & ( D, E, \{ \{ A_{ax} \}_x \}_a , 
\{ \{ B_{by} \}_y \}_b , \Psi )
\end{eqnarray}
such that $D, E$ are finite dimensional Hilbert spaces,
$\{ \{ A_{ax} \}_x \}_a$ is a family of $\mathcal{X}$-valued
positive operator valued measures (POVMs) on $D$ (indexed by $\mathcal{A}$),
$\{ \{ B_{by} \}_y \}_b$ is a family of $\mathcal{Y}$-valued
positive operator valued measures on $E$,
and $\Psi$ is a density operator on $D \otimes E$. In this paper, we assume without loss of generality that $\Psi$ is pure, written as $\Psi = \left| \psi \right> \left< \psi \right|$, and that the operators $A_{ax}$ and $B_{by}$ are all projectors.
We say that the
strategy $\Gamma$ \textit{achieves} the $2$-player correlation
$(P(xy|ab) )$ if $P(xy|ab) =  \Tr [ \Psi ( A_{ax} \otimes B_{by} ) ]$
for all $a, b, x, y$.  A correlation is a \textit{quantum} correlation if it can be achieved by such a $2$-player strategy.


\section{Local randomness from the NPA hierarchy}
\label{sec:bound}

The goal of this section is to derive an upper bound on \B's probability of guessing \A's after playing the CHSH game with her. The method we use is based on the Navascues-Pironio-Acin hierarchy which is introduced in the next subsection.
\subsection{Navascues-Pironio-Acin hierarchy}
\label{subsec:npa}

The Navascues-Pironio-Acin hierarchy, or NPA hierarchy, was introduced to characterize quantum correlations.  We briefly sketch the idea behind the hierarchy
and refer the reader to \cite{Navascues:2008} for the formal treatment.  The NPA hierarchy is an infinite series
of conditions which must be satisfied by any quantum correlation.


In the measurement scenario, we assume Alice and Bob share state $\ket{\psi}$ and will apply some measurements determined by the inputs.  For compatibility with \cite{Navascues:2008}, we use a different notation in this section and assume that each output letter is associated to a unique input letter --- i.e., 
each output letter $x \in \mathcal{X}$ is uniquely associated to a single input $A(x)$.  If Alice is given input $a$,
then her only valid outputs are those for which $a = A ( x )$.

A \textit{behavior} $P$ in this measurement scenario is a set of nonnegative values $P = \{P(x,y)\;:\quad x \in \calX, y \in \calY\}$
such that $\sum_{x \in a, y \in b} P ( x, y ) = 1$ for any $a \in \calA, b \in \calB$.
The definition of a quantum behavior is as follows.  (As we will discuss, it is somewhat different from the definition of quantum correlation.)
\theoremstyle{definition}
\begin{definition}
\label{quantdef}
	A behavior $P$ is a quantum behaviour if there exists a pure state $\ket{\psi}$ in a Hilbert space $\mathcal{H}$, a set of measurement operators $\{ E_x: x \in \mathcal{X}\}$ for Alice, and a set of measurement operators $\{ E_y: y \in \calY\}$ for Bob, such that $\forall x \in \calX$ and $\forall y \in \calY$
	\begin{align}
		P(x,y) = \bra{\psi}E_xE_y\ket{\psi},
	\end{align}
	with the measurement operators $E$ satisfying
	\begin{enumerate}
		\item $E_x^\dagger = E_x$ and $E_y^\dagger = E_y$,
		\item $E_xE_{\bar{x}} = \delta_{x\bar{x}} E_x$ if $A(x) = A(\bar{x})$ and $E_yE_{\bar{y}} = \delta_{y\bar{y}} E_y$ if $B(y) = B(\bar{y})$,
		\item $\sum_{x \in A^{-1}(a)} E_x = \bbI$ and $\sum_{y \in B^{-1}(b)} E_y = \bbI$ for all $a$, and
		\item $[E_x, E_y] = 0$.
	\end{enumerate}
\end{definition}
The first three properties ensure that the operators $E_x$ and $E_y$ are projectors and define proper measurements. The fourth property ensures that the measurements by \A\; and \B\; do not interfere with one another.
This definition is similar to the definition of a quantum correlation, but is based on commutativity
rather than bipartiteness.  Under these definitions, every quantum correlation yields
a quantum behavior (i.e., by setting $E_x = A_{ax} \otimes \mathbb{I}$, $E_y = \mathbb{I}
\otimes B_{by}$) but not necessarily vice versa \cite{Slofstra:2016}.

The idea of the hierarchy is that if we let $\mathcal{O}$ be any finite
set of operators that can be expressed as finite products of elements of the set $\{ E_x \}_x \cup \{ E_y \}_y$ (for example, $E_x$  or $E_x E_y E_{y'}$), then the matrix $\Gamma$ given by
\begin{align}
\Gamma_{ij} = \bra{\psi}O_i^\dagger O_j \ket{\psi}
\end{align}
where $O_i, O_j$ vary over the elements of $\mathcal{O}$, must be positive semidefinite.
Additionally, there are some independent equalities (which depend on the setting) that must be satisfied by the entries of $\Gamma$.  

We define a sequence of such matrices (certificates) as follows.
Since some of the $O_i$'s can be expressed in multiple ways as products of operators from $\{ E_x \}_x \cup \{ E_y \}_y$, we define the \textit{length} of the operator to be the minimum number of projectors needed to generate it. For
any $k \geq 1$, the $k$th \textit{certificate matrix} $\Gamma^{(k)}$ is the matrix associated to the set
$\mathcal{O}$ of all operators of length at most $k$.  The fact that $\Gamma^{(k)}$ must be positive semidefinite
constrains the possible entries in $\Gamma^{(k)}$, and in particular constrains the values
$P ( x, y ) = \left< \psi \mid E_x E_y \mid \psi \right>$ which can occur in a quantum behavior.
Thus we obtain a hierarchy of constraints on the set of all quantum behaviors.

Measuring the amount of local randomness after a nonlocal game is not as simple as constraining quantum behaviors (Definition~\ref{quantdef}) since
in particular, measurements that Bob uses to guess Alice's output may not commute with the measurements he used to play the game.  Fortunately,
the NPA hierarchy can also be adapted to scenarios which involve sequential measurements \cite{budroni2013bounding, pironio2010convergent}.  
In the next subsection, we apply an adaptation of the NPA hierarchy to study local randomness for the CHSH game.

\subsection{Application of the NPA hierarchy}
\label{subsec:npa_app}
The CHSH game is
defined on alphabets $\mathcal{X} = \mathcal{Y} = \mathcal{A} = \mathcal{B} = \{ 0, 1 \}$, and the input probability
\begin{eqnarray}
q ( a, b ) &= & 1/4
\end{eqnarray}
for all $a, b$. The score function is
\begin{eqnarray}
H( a, b, x, y ) & = & x \oplus y \oplus \neg ( a \wedge b ).
\end{eqnarray}
for all $a, b, x, y$.

As usual, we assume that Alice and Bob share some pure state $\ket{\psi}$. First, Alice gets input $a \in \calA$ and outputs $x \in \calX$. Bob gets input $b \in \calB$ and outputs $y \in \calY$. Then, Bob gets \A's input $a$ and outputs $x' \in \calX$. Alice's projective measurement for input $a$ and output $x$ is $A_{ax}$. Similarly, the projective measurement operator for input $b$ and output $y$ is $B_{by}$. To guess Alice's output, Bob's projective measurement is $B'_{abx'}$ after he gets Alice's input $a$ and outputs $x' \in \calX$.

In the semidefinite programming instance, the objective value is Bob's guessing probability, denoted by $P_2$. The constraints include the expression of $P_1$ and the commutation relations. Both $P_1$ and $P_2$ can be expressed by $A_{ax}$, $B_{by}$ and $B'_{abx'}$. The expressions can be found in Appendix \ref{app:exp}.
We use the third-order certificate to maximize $P_2$ for a given $P_1$ and get the following data.

The $P_2$ values are $1$, $0.995645$, $0.977018$, $0.95783$, $0.938371$, $0.918742$, $0.898992$, $0.879149$ and $0.859229$ when $P_1$ is ranging from $0.75$ to $0.85$ (see Figure~\ref{sandwichplot}). These points indicate the proved upper bound on Bob's guessing probability. Next, we derive a lower bound on $P_2$ to show how close the upper bound is to the actual optimal guessing probability.


First note that the optimal strategy for CHSH involves Alice and Bob sharing a Bell state $\ket{\Phi^+} = \frac{1}{\sqrt{2}} \left( \left| 00 \right> + \left| 11 \right> \right)$,
and Alice performing the $X$ or $Z$ measurement when her input is $0$ or $1$, respectively, and Bob performing the $(X+Z)/\sqrt{2}$ or $(X-Z)/\sqrt{2}$ measurement
when his input is $0$ or $1$, respectively.  This strategy achieves a score of $\frac{1}{2} + \frac{\sqrt{2}}{4}$ at CHSH, and moreover Bob can guess Alice's output given her input with probability $\frac{1}{2} + \frac{\sqrt{2}}{4}$, by simply guessing $x \oplus ( a \wedge b )$.
 
Consider the scenario where Alice and Bob share a random coin $R$. With probability $r$ or $1-r$, the coin $R$ has value $0$ or $1$, respectively. 
If $R = 0$, then Alice and Bob always output $0$, and if $R = 1$, then Alice and Bob play the optimal CHSH strategy.  In the former case, Bob can perfectly
guess Alice's output, while in the latter case, he can guess her output with probability $\frac{1}{2} + \frac{\sqrt{2}}{4}$.

Therefore, the expressions of $P_1$ and $P_2$ in terms of $r$ for this strategy are
\begin{align}
	P_1(r) &= \frac{3}{4} r + \frac{2+\sqrt{2}}{4}(1-r)\\
	P_2(r) &= 1\cdot r + \frac{2+\sqrt{2}}{4}(1-r).
\end{align}
Then the expression of $P_2$ in terms of $P_1$ is
\begin{align}
	P_2 = 1 + \frac{3\sqrt{2}}{4} - \sqrt{2} P_1.
\end{align}

To generate the plot in Figure~\ref{sandwichplot}, we plot the lower bound first. Then we mark the proved data points of the upper bound and connect then with dashed lines to indicate the approximate shape of the upper bound. For the upper bound point above $1$, we cut it off by the line $y=1$.
\begin{figure}[H]
	\centering
 	\includegraphics[width=0.5\textwidth]{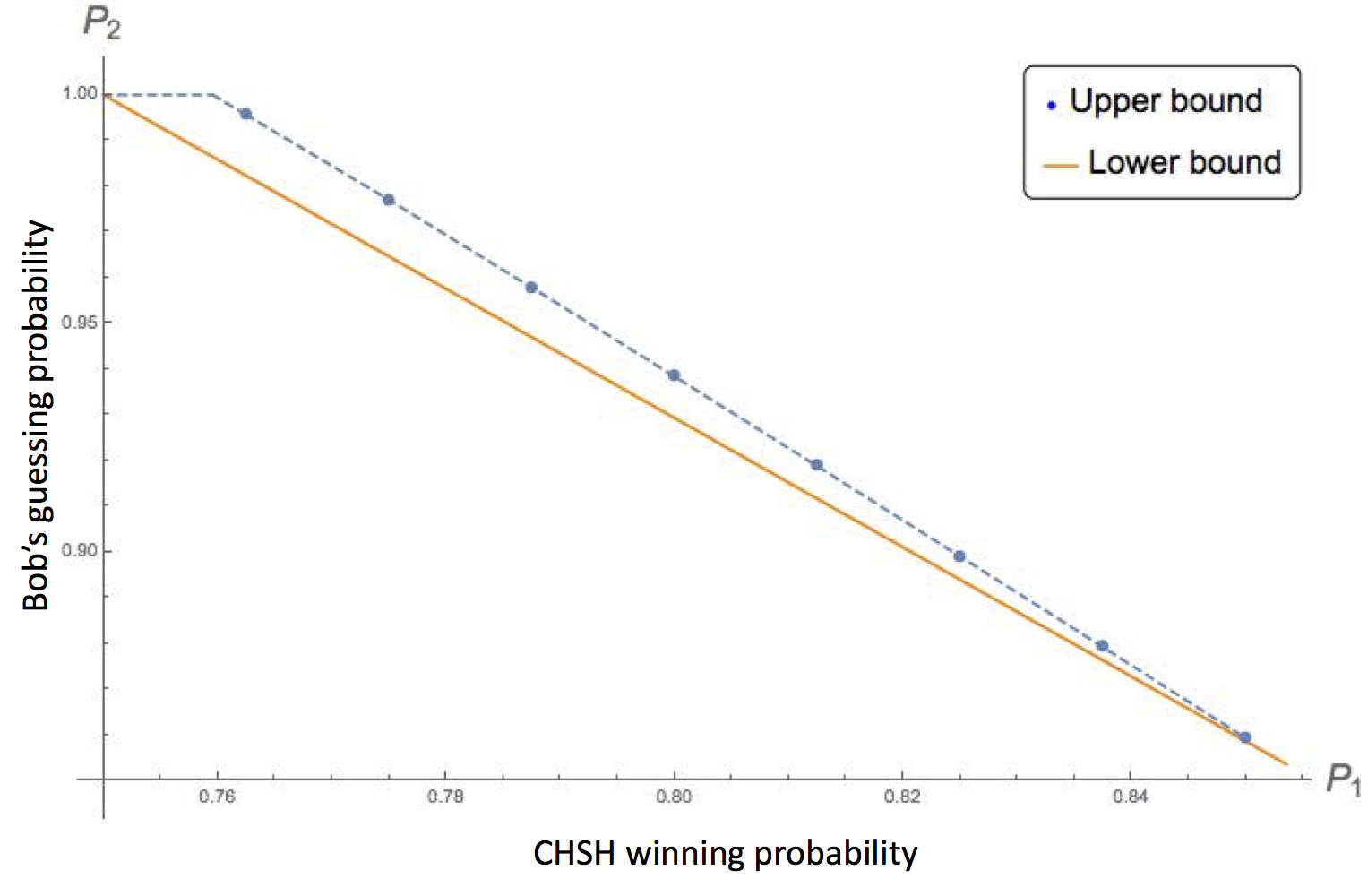}
	\caption{Plot of the lower and approximate upper bounds of $P_2$ against $P_1 \in (0.75, 0.85)$. }
\label{sandwichplot}
\end{figure}
The optimal (blind) rate curve for CHSH must lie in between the orange and blue curves
in Figure~\ref{sandwichplot}.

\section{Local randomness from rigidity}

\label{sec:rigid}

For games with larger alphabets than the CHSH game, using the above adaptation of the NPA hierarchy
is more difficult because of the size of the certificates.  In the current section we explore how techniques from
quantum rigidity can be used to prove blind rate curves.
The approach in the current section requires less
computation than the NPA hierachy approach, and although the rate curve we achieve
lacks the near-optimal properties of our 
rate curve for CHSH (Figure~\ref{sandwichplot}), it is
optimal as the score threshold approaches the optimal quantum score.

We study the Magic Square game, which, like CHSH, is a game with two players, Alice and Bob.  The input alphabets for Alice and Bob are $\mathcal{A} = \mathcal{B} = \{ 0, 1, 2 \}$, the input distribution $q$ is uniform, 
and the output alphabets are the sets of bit strings $\mathcal{X} = \{  000, 011, 101, 110 \}$ for Alice and $\mathcal{Y} = \{ 100, 010, 001, 111 \}$ for Bob.
The game is won if the inputs $a, b$ and outputs $x, y$ satisfy $x_b = y_a$, meaning that the $b$-th bit of $x$ equals the $a$-th bit of $y$.

A strategy for the Magic Square game consists of a pure state $\ket{\psi} \in \mathcal{H}_A \otimes \mathcal{H}_B$, and projective measurement families
$\left\{ \left\{ A_{ax} \right\}_x \right\}_a$ on $\mathcal{H}_A$ and $\left\{ \left\{ B_{by} \right\}_y \right\}_b$ on $\mathcal{H}_B$.  Note that we can let
\begin{eqnarray}
F^z_{ab} & = & \sum_{x_b = z} A_{ax} \\
G^z_{ab} & = & \sum_{y_a = z} B_{by} \\
F_{ab} & = & F^0_{ab} - F^1_{ab} \\
G_{ab} & = & G^0_{ab} - G^1_{ab},
\end{eqnarray}
and then the measurements will satisfy
\begin{eqnarray}
\prod_b F_{ab} & = & I \\
\prod_a G_{ab} & = & -I \\
F_{ab} F_{ab'}  & = & F_{ab'} F_{ab}  \\
G_{ab} G_{a'b} & = & G_{a'b} G_{ab} \\
F_{ab}^2 & = & I \\
G_{ab}^2 & = & I.
\end{eqnarray}
The measurement operators $A_{ax}$ and $B_{by}$ can
be recovered from $\{ F_{ab} \}, \{ G_{ab} \}$, and thus to specify a strategy it suffices to specify $\ket{\psi}, \{ F_{ab} \}, \{ G_{ab} \}$ satisfying
the above conditions.  We refer
to the triple $\left( \ket{\psi}, \{ F_{ab} \}, \{ G_{ab} \} \right)$ as a \textit{reflection strategy} for the Magic Square game.  

Suppose that a reflection strategy $\left( \ket{\psi}, \{ F_{ab} \}, \{ G_{ab} \} \right)$
achieves a score of $1-\delta$.  Appendix~\ref{app:rig} proves the following inequalities
for any $a, a', b, b' \in \{ 0, 1 , 2 \}$ with $a \neq a', b \neq b'$, using steps from the proof of rigidity for the Magic Square game \cite{WuBancal:2016_p}:
\begin{eqnarray}
\label{rigineq1}
\left\| F_{ab} \otimes G_{ab} \left| \psi \right> - \left| \psi \right> \right\| & \leq & 6 \sqrt{\delta} ~~~~~~ \\
\label{rigineq2}
\left\| F_{ab} F_{a' b'} \otimes I \left| \psi \right> + F_{a'b'} F_{a b} \otimes I \left| \psi \right> \right\| & \leq & 6 \sqrt{\delta} ~~~~~~
\end{eqnarray}
The next proposition uses the above inequalities to prove that in a high-performing strategy, if Alice measures with $F_{ab}$ and Bob measures
with $G_{a'b'}$, with $a \neq a', b \neq b'$, then the outcome of Alice's measurement is nearly undetectable to Bob.
\begin{proposition}
\label{prop:dist}
Let $a, a', b, b' \in \{ 0, 1, 2 \}$, $z \in \{ 0, 1 \}$
be such that $a \neq a', b \neq b'$.
Let $\left( \ket{\psi}, \{ F_{ab} \}, \{ G_{ab} \} \right)$ be a reflection strategy for the Magic Square game which
achieves an expected score of $1 - \delta$.  Then, the post-measurement states
\begin{eqnarray}
\label{guessstate1}
\Tr_A \left[ (F_{ab}^0  \otimes G_{a'b'}^z ) \ketbra{\psi}{\psi} (F_{ab}^0  \otimes G_{a'b'}^z ) \right]
\end{eqnarray}
and
\begin{eqnarray}
\label{guessstate2}
\Tr_A \left[ (F_{ab}^1  \otimes G_{a'b'}^z ) \ketbra{\psi}{\psi}(F_{ab}^1  \otimes G_{a'b'}^z ) \right]
\end{eqnarray}
are separated by trace distance at most $18 \sqrt{\delta}$.
\end{proposition}

\begin{proof}
Applying inequality (\ref{rigineq2}), we have the following, in which we use the notation
$u =_x v$ to denote that the Euclidean distance between the vectors $u$ and $v$ is no more than $x$:
\begin{eqnarray*}
F_{a'b'} F_{ab}^0 \otimes I \left| \psi \right>
& = & 
F_{a'b'} \left( \frac{ I + F_{ab} }{2} \right) \otimes I \left| \psi \right> \\
& ~~ =_{3 \sqrt{\delta}} ~~ & 
\left( \frac{ I - F_{ab} }{2} \right) F_{a'b'}  \otimes I \left| \psi \right> \\
& = & 
F_{ab}^1 F_{a'b'}  \otimes I \left| \psi \right>.
\end{eqnarray*}
Therefore,
\begin{eqnarray*}
F_{a'b'} F_{ab}^0 \otimes G_{a'b'}^z \left| \psi \right> 
& ~~ =_{3 \sqrt{\delta}} ~~ & 
 F_{ab}^1 F_{a'b'} \otimes G_{a'b'}^z \left| \psi \right> \\
& =_{6 \sqrt{\delta}} & 
 F_{ab}^1  \otimes G_{a'b'}^z G_{a'b'} \left| \psi \right>  \\
& = & 
(-1)^z F_{ab}^1  \otimes G_{a'b'}^z \left| \psi \right> 
\end{eqnarray*}
Therefore, since $\left\| u u^* - v v^* \right\|_1  \leq 2 \left\| u - v \right\|$
for any unit vectors $u,v$, we find that the trace distance between the projectors
\begin{eqnarray}
(F_{a'b'} F_{ab}^0  \otimes G_{a'b'}^z ) \ketbra{\psi}{\psi} (F_{ab}^0 F_{a'b'}  \otimes G_{a'b'}^z )
\end{eqnarray}
and 
\begin{eqnarray}
(F_{ab}^1  \otimes G_{a'b'}^z ) \ketbra{\psi}{\psi} (F_{ab}^1  \otimes G_{a'b'}^z )
\end{eqnarray}
is upper bounded by $18 \sqrt{\delta}$.    Applying the partial trace
over $\mathcal{H}_A$ to both projectors (and dropping the $F_{a'b'}$ terms, which become irrelevant),
we obtain the desired result.
\end{proof}

The next corollary follows easily.

\begin{corollary}
\label{cor:magicguess}
Let $\left( \ket{\psi}, \left\{ \left\{  A_{ax} \right\}_x \right\}_a, \left\{ \left\{ B_{by} \right\}_y \right\}_b \right)$ be a strategy for the Magic Square game which achieves
an expected score of $1 - \delta$.  Let $a, b, b' \in \{ 0, 1, 2 \}$ be such that $b \neq b'$, and suppose that the strategy is executed on inputs $a,b$ and outputs $x,y$
are obtained.  Then the probability that Bob can subsequently guess $x_{b'}$ given $b'$ is no more than $\frac{1}{2} +
9  \sqrt{\delta}$.
\end{corollary}

\section{The deletion certification protocol}
\label{sec:cert}

We next focus on the problem of certified deletion, which we describe as follows.  Alice wishes to interact with an untrusted device ($D^a$) and a
second party (Bob) so as to prepare for herself a random bit $m$ and a classical string $k$, such that after the interaction
is complete the following conditions hold:
\begin{enumerate}
\item[(A)] If Alice were to give $k$ to Bob immediately, then Bob could recover the bit $m$.

\item[(B)] There is a \textit{deletion procedure} that Alice and Bob can carry out,
involving classical communication only, such that
after the protocol is over Bob will not be able to recover $m$ even if he were given $k$.
\end{enumerate}
Note that this procedure can be used as a form of encryption: if Alice has a predetermined
secret message bit $y \in \{ 0, 1 \}$ which she wishes to encrypt, then she can execute the same preparation procedure and then transmit
the XOR bit $y \oplus m$ to Bob.  Recovering or deleting $y$ is then equivalent to recovering
or deleting $m$.

Variants of this problem have been studied in other settings (e.g., \cite{Unruh:2015} in 
a computational setting, \cite{Rib:2016, Kaniewski:2016} in a bounded storage model). 
Our setting is the \textit{device-independent} setting, where the honest
user Alice does not trust the quantum processes used
in the protocol.  Our protocol is based on the Magic Square game.  We make the following assumptions:
\begin{enumerate}
\item Alice and Bob possess an untrusted $2$-part device $D = (D^a, D^b )$
which is compatible with the Magic Square game.

\item Alice has the ability to generate private (trusted) randomness.

\item Alice's device $D^a$ does not communicate information to Bob or to $D^b$ once
the protocol is underway.

\item Alice and Bob have the ability to communicate classically.
\end{enumerate}
No assumptions are made about Bob's behavior --- in particular, he may perform
arbitrary operations on any quantum information that is contained inside of the device $D^b$
that he possesses.\footnote{We could model Bob's behavior simply by allowing him to possess a quantum system $Q$ and
to perform arbitrary operations on it.  We have chosen to allow him to have a device 
because it is easier to express his behavior in the case where he is honest.}


It is helpful to change notation from the previous section.
The protocol will contain two sequence of inputs to the Magic Square game, one for Alice and one for Bob,  which will be denoted by $\pmb{v}^a = (v_1^a, v_2^a \dots v_N^a)$ and $\pmb{v}^b = (v_1^b, v_2^b \dots v_N^b)$. The sequences of the outputs will be $\pmb{h}^a = (h_1^a, h_2^a \dots h_N^a)$ for Alice and $\pmb{h}^b = (h_1^b, h_2^b \dots h_N^b)$ for Bob. The initial preparation protocol is given Figure~\ref{certprot}.

\begin{figure}[H]
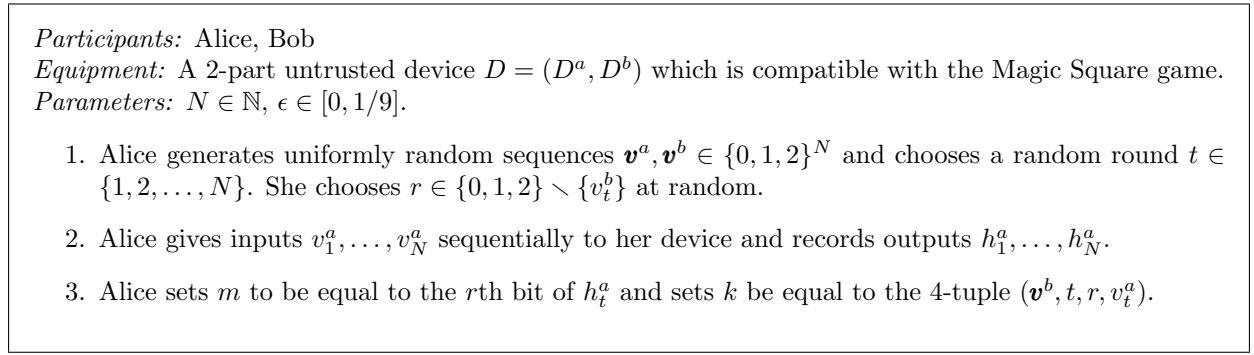

\begin{framed}
\textit{Participants:} Alice, Bob \\
\textit{Equipment:} A $2$-part untrusted device $D = (D^a, D^b)$ which is compatible with the Magic Square game. \\
\textit{Parameters:} $N \in \mathbb{N}$, $\epsilon \in [0, 1/9]$.
\begin{enumerate}
	\item Alice generates uniformly random sequences $\pmb{v}^a, \pmb{v}^b \in \{ 0, 1, 2 \}^N$ and chooses a random round
	$t \in \{ 1, 2, \ldots, N \}$.  She chooses $r \in \{ 0, 1 , 2 \} \smallsetminus \{ v^b_t \}$ at random.
	\item Alice gives inputs $v^a_1, \ldots, v^a_N$ sequentially to her device and records outputs $h^a_1 , \ldots , h^a_N$.
	\item Alice sets $m$ to be equal to the $r$th bit of $h^a_t$ and sets $k$ be equal to the $4$-tuple $(\pmb{v}^b, t, r, v_t^a)$.
\end{enumerate}
\end{framed}
\caption{The preparation protocol ($PREP$)}
\label{certprot}
\end{figure}

We wish to show first that it is possible for Bob to determine $m$ if he were given $k$.  This is straightforward:
if the device $D = (D^a, D^b )$ were such that it wins the Magic Square game with probability
$1 - \epsilon$ at each use, then the protocol
in Figure~\ref{recprot} successfully determines $m$ with probability $1 - \epsilon$.

\begin{figure}[H]
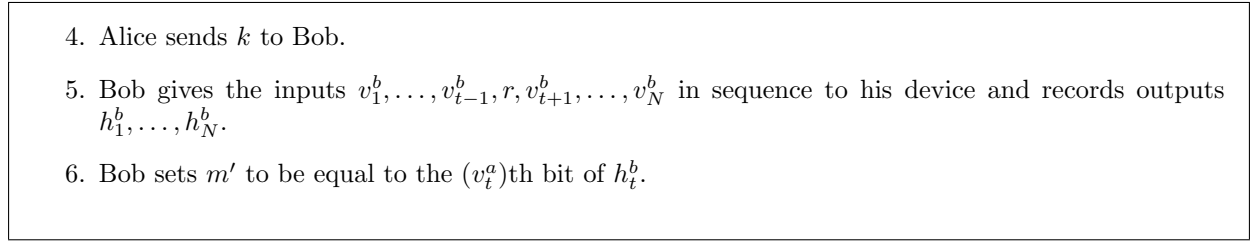

\begin{framed}
\begin{enumerate}
	\item[4.] Alice sends $k$ to Bob.
	\item[5.] Bob gives the inputs $v^b_1,  \ldots, v^b_{t-1}, r, v^b_{t+1}, \ldots, v^b_N$ in sequence
	to his device and records outputs $h^b_1, \ldots, h^b_N$.
	\item[6.] Bob sets $m'$ to be equal to the $(v_t^a)$th bit of $h^b_t$.
\end{enumerate}
\end{framed}
\caption{The recovery protocol ($REC$)}
\label{recprot}
\end{figure}

Next we wish to show that there is a protocol which makes $m$ unrecoverable for Bob (even while it allows
Bob to know the key $k$ after the protocol is completed, and allows him to have access to all remaining
quantum information in the device $D^b$).  We use the protocol $DEL$ in Figure~\ref{delprot}, which is also meant to follow
the protocol $PREP$ in Figure~\ref{certprot}.  The protocol has Bob play his side of the Magic Square
game and then has Alice check the resulting score.  Then at the conclusion of the protocol, Alice reveals
the key $k$ to Bob (which is merely a convenience for stating the security of the protocol).

\begin{figure}[H]
\begin{framed}
\begin{enumerate}
	\item[4.] For $i = 1, 2, \ldots, N$, Alice sends Bob the input $v^b_i$ and Bob sends back
	an output $h^b_i$. 
	\item[5.] Alice computes the average score at the Magic Square game (across $N$ rounds) achieved by the input sequences $\pmb{v}^a,
	\pmb{v}^b$ and output sequences $\pmb{h}^a, \pmb{h}^b$.  If this average is greater than or equal to $1 - \epsilon$, she accepts Bob's responses; otherwise, she aborts the protocol.
\item[6.] Alice sends $k$ to Bob.
\end{enumerate}
\end{framed}
\caption{The deletion protocol ($DEL$)}
\label{delprot}
\end{figure}

Note that at step 4 in Figure~\ref{delprot}, the interactions must be done in sequence (i.e., Alice waits
to receive $h^b_{i}$ before revealing $v^b_{i+1}$).  Bob can use his device $D^b$ to obtain
his outputs, but we do not require that.

The following theorem asserts the security of the deletion protocol $DEL$.  Let $SUCC$ denote the event that Alice ``accepts''  at step 5 in Figure~\ref{delprot}.
\begin{theorem}
\label{thm:cert}
	\label{thm:cdp}
	Assume that $P ( SUCC ) > 0$ in protocol $DEL$.  Then, the probability that Bob can guess $m$ at the conclusion of the protocol, conditioned on $SUCC$,
	is upper bounded by
	\begin{eqnarray}
         \label{theuppbound}
         \frac{1}{2} + 9 \sqrt{\epsilon + N^{-1/4}}  + \frac{e^{-\sqrt{N }/2}}{P ( SUCC ) }.
	\end{eqnarray}
\end{theorem}

Note that if we fix a constant $\gamma > 0$, assume that $P ( SUCC ) > \gamma$, and let $N$ tend to infinity,
then the upper bound (\ref{theuppbound}) tends to $\frac{1}{2} + 9 \sqrt{\epsilon}$.

For the proof of Theorem~\ref{thm:cert}, we will need the following lemma.  

\begin{lemma}
	\label{cor:win}
Let $I_i$ denote indicator variable for the event that the $i$th round is won.  Let
\begin{eqnarray}
I'_i & = & E ( I_i \mid I_{i-1} I_{i-2} \cdots I_1 ),
\end{eqnarray}
and let $\overline{I}' = ( \sum_i I'_i )/N$.  Then for any $\mu > 0$,
\begin{eqnarray}
\label{critdiff}
Pr ( SUCC \wedge (\overline{I}' < 1 - \epsilon - \mu )) & \leq & e^{- \frac{N \mu^2}{2}}.
\end{eqnarray}
\end{lemma}

\begin{proof}
Let $\overline{I} = ( \sum_i I_i)/N$.
Let
	\begin{align}
		Z_i = \sum_{j=1}^i(I_j - I_j').
	\end{align}
	Then $\{Z_0, Z_1,\dots,Z_N\}$ is a martingale:
	\begin{align*}
		E(Z_{i+1}|Z_i,\dots,Z_1) = Z_i + E(I_{i+1}|I_{i} \cdots I_1) - I_{i+1}' = Z_i.
	\end{align*}
Therefore by Azuma's inequality, the probability of the event $\sum_i (Z_i ) > \mu$ is upper bounded by $e^{- \frac{N \mu^2}{2}}$.
The event in inequality (\ref{critdiff}) implies $\sum_i (Z_i ) > \mu$, and the desired result follows.
\end{proof}

Now we can prove the main theorem of this section.
\begin{proof}[Proof of Theorem \ref{thm:cdp}]
By Corollary~\ref{cor:magicguess}, for any $i$ and any $c \in \{ 0, 1, 2 \} \smallsetminus v^b_i$, the probability that 
Bob can guess the $c$th bit of $h^a_i$ is upper bounded by $\frac{1}{2} + 9 \sqrt{ 1 - I'_i }$.  Therefore, the probability
that Bob can guess $m$ at the conclusion of the protocol $DEL$ is no more than
\begin{eqnarray}
\left[ \sum_{i=1}^N  \left( \frac{1}{2} + 9 \sqrt{ 1 - I'_i }  \right) \right] / N,
\end{eqnarray}
which by the concavity of the square root function is upper bounded by
\begin{eqnarray}
\frac{1}{2} + 9 \sqrt{ 1 - \overline{I}'},
\end{eqnarray}
For any $\mu > 0$, we have by Lemma~\ref{cor:win},
\begin{eqnarray*}
Pr [ \overline{I}' \geq 1 - \epsilon - \mu \mid SUCC ] \geq 1 - \frac{ e^{-N \mu^2/2}}{Pr ( SUCC)},
\end{eqnarray*}
and therefore, conditioned on $SUCC$, Bob's probability of guessing $m$ is upper bounded by
\begin{eqnarray}
\frac{1}{2} + 9 \sqrt{ \epsilon + \mu}  + \frac{e^{-N \mu^2/2}}{Pr ( SUCC)}.
\end{eqnarray}
Setting $\mu = N^{-1/4}$ yields the desired result.
\end{proof}
\vskip0.2in
\noindent\textbf{Acknowledgements.} This work includes contributions from the National Institute of Standards and Technology and is not subject to U.S. copyright. This research was supported in part by NSF grant 1526928.

\bibliographystyle{plain}
\bibliography{local}

\appendix

\newpage
\section{Expressions of $P_1$ and $P_2$}
\label{app:exp}
The winning probability of the CHSH game is 
\begin{equation*}
\begin{aligned}
	P_1 = 1/4(P(00|00) + P(11|00) 
	+P(00|01) + P(11|01) \\
	+P(00|10) + P(11|10)
	+P(01|11) + P(10|11)),
\end{aligned}
\end{equation*}
where $Pr(xy|ab) = \bra{\psi} A_{ax}B_{by} \ket{\psi}$. Since for any input $a$ and $b$, $A_{a1} = \bbI - A_{a0}$ and $B_{b1} = \bbI - B_{b0}$, we can express $P_1$ in terms of the projectors as 
\begin{equation}
	\begin{aligned}
	P_1 = \bra{\psi} \big( \frac{3}{4} - \frac{1}{2} A_{00} -\frac{1}{2}B_{00}  + \frac{1}{2} A_{00}B_{00} \\
	+ \frac{1}{2} A_{00}B_{10} + \frac{1}{2} A_{10}B_{00} - \frac{1}{2} A_{10}B_{10}\big)\ket{\psi}.
	\end{aligned}
\end{equation}

When Bob wants to guess Alice's output $x$ given $a$ and $b$, the probability that he can guess correctly is
\begin{equation}
	\begin{aligned}
	P_2 =  1/4 \sum_{b,y} \big(Pr(0y0|0b)+ Pr(1y1|0b)\\
	+Pr(0y0|1b) + Pr(1y1|1b)\big)
	\end{aligned}	
\end{equation}
where 
\begin{equation*}
\begin{aligned}
 Pr(xyx'|ab) &=  \bra{\psi}A_{ax}^\dagger B_{by}^\dagger B_{abx'}^{'\dagger} B'_{abx'}B_{by}A_{ax}\ket{\psi}\\
 		&= \bra{\psi} A_{ax}^\dagger B_{by}^\dagger B'_{abx'}B_{by} \ket{\psi}.
\end{aligned}
\end{equation*}
The measurement $\{ \{ B'_{abx'} \}_{x'} \}_{ab}$ is a set of measurements indexed by $(a, b) \in \mathcal{A} \times \mathcal{B}$.\footnote{Note that it not necessary to make Bob's second measurement depend on the outcome of his first measurement, since
that outcome ($y$) is recoverable from the postmeasurement state of his first measurement.}
The two measurements $\{ \{ B_{by} \}_y \}_b$ and $\{ \{ B'_{abx'} \}_{x'} \}_{ab}$ commute with $\{ \{ A_{ax} \}_x \}_a$.

The probability $P_2$ can be expressed in terms of the projectors as $P_2 = \frac{1}{4} \bra{\psi} S \ket{\psi}$ with $S$ defined as 
\begin{equation}
\begin{aligned}
	\frac{1}{4} S = & \bbI - \frac{1}{2} (A_{00} + A_{10} ) - \frac{1}{4} (B'_{000} + B'_{010} + B'_{100} + B'_{110})\\
		+&\frac{1}{2} ( A_{00}B'_{000} + A_{00}B'_{010} + A_{10}B'_{100} + A_{10}B'_{110})\\
		+&\frac{1}{4} (B_{00}B'_{000} + B'_{000}B_{00} + B_{10}B'_{010} + B'_{010}B_{10} + B_{00}B'_{100} + B'_{100}B_{00} + B_{10}B'_{110} + B'_{110}B_{10})\\
		- &\frac{1}{2} (A_{00}B_{00}B'_{000} + A_{00}B'_{000}B_{00} + B_{00}B'_{000}B_{00})
		- \frac{1}{2} (A_{00}B_{10}B'_{010} + A_{00}B'_{010}B_{10} + B_{10}B'_{010}B_{10})\\
		- &\frac{1}{2} (A_{10}B_{00}B'_{100} + A_{10}B'_{100}B_{00} + B_{00}B'_{100}B_{00})
		- \frac{1}{2} (A_{10}B_{10}B'_{110} + A_{10}B'_{110}B_{10} + B_{10}B'_{110}B_{10})\\
		+& A_{00}B_{00}B'_{000}B_{00} + A_{00}B_{10}B'_{010}B_{10} + A_{10}B_{00}B'_{100}B_{00} + A_{10}B_{10}B'_{110}B_{10}. 
\end{aligned}
\end{equation}
Here we use the relation $B'_{ab1} = \bbI - B'_{ab0}$ again. 
\section{Proof of Inequalities (\ref{rigineq1})--(\ref{rigineq2})}
\label{app:rig}
We follow steps from the proof of rigidity for the Magic Square game in \cite{WuBancal:2016_p}.  (See also \cite{Jain:2017},
which performs a similar derivation based on \cite{WuBancal:2016_p}.)
By symmetry, it suffices to address the single case where $a = b = 0, a' = b' = 1$, so we will assume
those values from now on.  Denote the probability that Alice and Bob lose the Magic Square game on
inputs $(i,j)$ by $\delta_{ij}$.  The average of these quantities over all $i, j \in \{ 0, 1, 2 \}$ is equal to $\delta$.  By linearity, we can compute the quantities $\delta_{ij}$ from the reflection strategy via the following expression:
\begin{eqnarray}
\left< \psi \right| F_{ij} \otimes G_{ij} \left| \psi \right> & = & 1 - 2 \delta_{ij}.
\end{eqnarray}
Therefore,
\begin{eqnarray*}
\left\|  F_{ij} \otimes G_{ij} \left| \psi \right>  - \left| \psi \right> \right\| & = & \sqrt{2 - 2 \left< \psi \right| F_{ij} \otimes G_{ij} \left| \psi \right>} \\
& = & 2 \sqrt{\delta_{ij}}, \\
\end{eqnarray*}
which proves (\ref{rigineq1}), since $2 \sqrt{\delta_{ij}} \leq 2 \sqrt{ 9 \delta } = 6 \sqrt {\delta}$.  

Let $\epsilon_{ij} = 2 \sqrt{\delta_{ij}}$.  We then
have the following, in which we let the expression $u =_x v$ denote that the Euclidean distance between
the vectors $u$ and $v$ is no more than $x$.  
\begin{eqnarray*}
F_{00} F_{11} \otimes I \left| \psi \right> & =_{\epsilon_{11}} & 
F_{00} \otimes G_{11} \left| \psi \right> \\
& = & 
- F_{02} F_{01} \otimes G_{21} G_{01} \left| \psi \right> \\
&=_{\epsilon_{01}}& - F_{02} \otimes G_{21} \left| \psi \right> \\
&=_{\epsilon_{02}}&  I \otimes G_{21} G_{02} \left| \psi \right> \\
&=& I \otimes G_{21} G_{22} G_{12} \left| \psi \right> \\
&=_{ \epsilon_{12}} & F_{12} \otimes G_{21} G_{22} \left| \psi \right> \\
&=_{\epsilon_{22}} & F_{12} F_{22} \otimes G_{21} \left| \psi \right> \\
&=_{\epsilon_{21} } & F_{12} F_{22} F_{21} \otimes I \left| \psi \right> \\
&=& F_{12} F_{20} \otimes I \left| \psi \right> \\
&=_{\epsilon_{20}} & F_{12}  \otimes G_{20} \left| \psi \right> \\
&= & - F_{11} F_{10}  \otimes G_{00} G_{10} \left| \psi \right> \\
&=_{\epsilon_{10}} & - F_{11} \otimes G_{00} \left| \psi \right> \\
&=_{\epsilon_{00}} & - F_{11} F_{00} \otimes I \left| \psi \right>
\end{eqnarray*}
Therefore, using the concavity of the square root function,
\begin{eqnarray*}
\left\| F_{00} \otimes G_{11} \left| \psi \right> + F_{11} \otimes G_{00} \left| \psi \right> \right\|
& \leq & \sum_{ij} \epsilon_{ij} \\
& = & 2 \sum_{ij} \sqrt{\delta_{ij}} \\
&= & 2 \cdot 9 \cdot \sum_{ij} \sqrt{\delta_{ij}}/9  \\
& \leq & 2 \cdot 9 \cdot \sqrt{ \sum_{ij} \delta_{ij}/9 } \\
& = & 6 \sqrt{ \delta },
\end{eqnarray*}
which implies (\ref{rigineq2}) as desired.
\end{document}